\documentclass[conference]{IEEEtran}

\makeatletter
\def\ps@headings{%
\def\@oddhead{\mbox{}\scriptsize\rightmark \hfil \thepage}%
\def\@evenhead{\scriptsize\thepage \hfil \leftmark\mbox{}}%
\def\@oddfoot{}%
\def\@evenfoot{}}
\makeatother
\usepackage{ifpdf}
\usepackage{amssymb}
\usepackage{color}
\usepackage{graphicx}
\usepackage{amsmath}

\usepackage{algorithm}
\usepackage{algorithmic}
\usepackage{multirow}

\newtheorem{theorem}{\textbf{Theorem}}[section]
\newtheorem{lemma}[theorem]{\textbf{Lemma}}

\newenvironment{definition}[1][Definition]{\begin{trivlist}
\item[\hskip \labelsep {\bfseries #1}]}{\end{trivlist}}

\usepackage{xspace}

\newcommand{\comment}[1]{}
\newcommand{\domain}{\mathcal{\mathbf{D}}}
\newcommand{\domainc}{\mathcal{\overline{\mathbf{D}}}}
\newcommand\length[1]{|#1|}
\newcommand{\bridge}{\mathcal{\mathbf{B}}}
\newcommand{\edge}{\mathcal{\mathbf{E}}}
\newcommand{\arc}{\vec{\mathcal{\mathbf{E}}}}
\newcommand{\vertex}{\mathcal{\mathbf{V}}}
\newcommand{\switch}{\mathcal{\mathbf{S}}}
\newcommand{\ugraph}{\mathcal{\mathbf{G_u}}}
\newcommand{\dgraph}{\mathcal{\mathbf{G_d}}}
\newcommand{\WO}{\mathcal{\mathbf{W_{opt}}}}
\newcommand{\WOnew}{\mathcal{\mathbf{W'_{opt}}}}
\newcommand{\WOrev}{\mathcal{\mathbf{W^{rev}_{opt}}}}
\newcommand{\W}{\mathcal{\mathbf{W}}}
\newcommand{\SetW}{\mathcal{\mathbf{\Omega}}}
\newcommand{\WOlen}{L_{opt}}
\newcommand{\Wlen}{L}
\newcommand{\Wlenk}{L_k}

\usepackage{cite}

\usepackage{algorithmic}
\usepackage{array}

\usepackage[tight,footnotesize]{subfigure}

\usepackage[font=footnotesize]{subfig}

\usepackage[normalem]{ulem}

\begin{document}

%
\title{On Diagnosis of Forwarding Plane via Static Forwarding Rules in Software Defined Networks}


\author{\IEEEauthorblockN{Ula\c{s} C. Kozat, Guanfeng Liang and  Koray K\"{o}kten}
\IEEEauthorblockA{DOCOMO Innovations, Inc., Palo Alto, CA 94304\\
Email: \{kozat,gliang,kkokten\}@docomoinnovations.com}
}



\maketitle

\begin{abstract}
Software Defined Networks (SDN) decouple the forwarding and control planes from each other. The control plane
is assumed to have a global knowledge of the underlying physical and/or logical network topology so that it can monitor,
abstract and control the forwarding plane. In our paper, we present solutions that install an optimal or near-optimal (i.e., within $14\%$ of the optimal) number of static forwarding rules on switches/routers so that any controller can verify the topology connectivity and detect/locate link failures at data plane speeds without relying on state updates from other controllers. Our upper bounds on performance indicate that sub-second link failure localization is possible even at data-center scale networks. For networks with hundreds or few thousand links, tens of milliseconds of latency is achievable. 
\end{abstract}

\section{Introduction}
SDNs are emerging as a principal component of future IT, ISP, and telco infrastructures. It promises to change networks from a collection of independent autonomous boxes to a well-managed, flexible, multi-tenant transport fabric \cite{McKeown:2008,McKeown:Keynote}.  As core principles, SDNs (i) decouple the forwarding and control plane, (ii) provide well-defined forwarding abstractions (e.g., pipeline of flow tables), (iii) present standard programmatic interfaces to these abstractions (e.g., OpenFlow), and (iv) expose high level abstractions (e.g., VLAN, topology graph, etc.) as well as interfaces to these service layer abstractions (e.g., access control, path control, etc.). \comment{Network controllers that are in charge of a given forwarding plane must know (ii) and implement items (iii) and (iv), accordingly.}

To fulfill its promise to convert the network to a well-managed fabric, presumably, a \emph{logically} centralized network controller is in charge of the whole forwarding plane in an end-to-end fashion with a global oversight of the forwarding elements and their inter-connections (i.e., nodes and links of the forwarding topology) on that plane. However, this might not be always true. For instance, there might be failures (software/hardware failures, buggy code, configuration mistakes, management plane overload, etc.) that disrupt the communication between the controller and a strict subset of forwarding elements. In another interesting case, the forwarding plane might be composed of multiple administrative domains under the foresight of distinct controllers. If controller of a given domain fails to respond or has very poor monitoring and reporting, then the other controllers might have a stale view of the overall network topology leading to suboptimal or infeasible routing decisions. 

For many systems that require a high grade of network availability and performance, it is essential for a controller to be able to verify the forwarding plane topology and identify link failures at forwarding plane speeds even under the aforementioned conditions. To attain this, we propose solutions that allocate a fraction of forwarding rules at each switch for control flows that can verify the topological connectivity, detect link failures, and locate one or more failed links. Our solutions are either optimal or order optimal in number of forwarding rules to be allocated for control flows as well as in number of control messages. Although our solutions are not optimal in latency in general, we can achieve tens of milliseconds of latency for topologies with 1000 links. For data-center scale topologies with 100K links, we present how one can effectively trade-off overhead in hardware rules and/or control messages to achieve sub-second latencies. Our solutions can guarantee locating an arbitrary single link failure, while it  can also probabilistically locate multiple failed links. We show over a simple example why one cannot guarantee locating multiple link failures over arbitrary topologies and failure scenarios.  We also present our simulation results over real topologies to quantify the overhead of our solution and its success of locating more than one failed link.

The paper is organized as follows. In Section~\ref{se:model}, we provide the detailed system model. In Section~\ref{se:lesstrivial}, we focus on verification of topology connectivity and establish optimality results as well as performance bounds. In Section~\ref{se:strong}, we turn our attention to locating a single but arbitrary link failure. We provide an order optimal solution and present performance bounds on important metrics. In Section~\ref{se:hard}, we extend the results to multiple link failures.  In Section~\ref{sec:results}, we show performance results using publicly available, real world topologies. In Section~\ref{se:related}, we cover the most related works. Finally, we conclude in Section~\ref{se:conc}.

\section{System Model}
\label{se:model}
\subsection{Network Architecture}
The main features of the network architecture is captured in Fig.~\ref{fig:systemmodel}. Our system model follows the OpenFlow model \cite{McKeown:2008}. Network consists of a forwarding plane and a control plane. The forwarding plane consists of forwarding elements (will use the term \emph{switch} interchangeably), each supporting a finite number of \emph{\{match, action\}} rules. A \emph{match} pattern is defined using incoming port ID and packet headers using ternary values 0, 1, or * (i.e., don't care). In essence, each match pattern defines a network flow. The following \emph{actions} on a particular flow match are allowed: \emph{forward to outgoing port ID}, \emph{drop packet}, \emph{pop outer header field}, \emph{push outer header field}, \emph{overwrite header field}. Regardless of match length and number of actions taken per matching (e.g., switch can first rewrite a particular field in the packet header, then push an MPLS label, and finally forward to a particular interface), we will count the cost of each flow matching rule as one forwarding rule.   The forwarding plane in Fig.~\ref{fig:systemmodel} has seven switches ($s_1$ through $s_7$) and nine interfaces/links between them. We assume that each link is bidirectional, i.e., a switch can both receive and send over the same link. The model is applicable to cases where multiple interfaces exist between the same pair of switches, there are logical interfaces  (e.g., a preconfigured tunnel with a tunnel ID configured on both ends), or pairs of directional links in opposite directions interconnect the same pair of switches.

The control plane consists of one or more controllers. \comment{In Fig.~\ref{fig:systemmodel},  there are three controllers.} A switch can be programmed (i.e., forwarding rules are installed) by only one controller (called \emph{master} controller). In contrast, a controller can install forwarding rules on multiple switches using control interfaces (depicted by red dotted lines in the figure). The control interfaces can be in-band (i.e., a slice of the forwarding plane are used for inter-controller and controller-switch communication) or out-of-band (i.e., a separate physical network interconnects controllers with each other and switches) or a mixture of both.  We refer to the subset of switches that a controller can send packets to and install rules on as its control domain $\domain$. The rest of the switches constitutes its complementary domain $\domainc$. For $C_3$, $\domain = \{s_4, s_7\}$ and $\domainc = \{s_1, s_2, s_3, s_5, s_6\}$, i.e., $C_3$ can send control packets to and install forwarding rules on $s_4$ and $s_7$.  To install a rule on a switch in its $\domainc$ (e.g., $s_1$), $C_3$ must send its request to the master controller (e.g., $C_1$) of that switch.

Switches are only allowed to check the health of their local forwarding interfaces and report to the controllers using control interfaces. Therefore, a controller cannot receive any failure notifications for the interfaces between the switches in its complementary domain, e.g., an interface failure between $s_2$ and $s_3$ is not reported to $C_3$ by $s_4$ or $s_7$. 

One way for a controller to learn about the failures within its $\domainc$ is to receive state updates from other controllers in charge of $\domainc$. This path however might be slow or might be disrupted for various reasons (e.g., configuration mistakes, controller overload, software/hardware failures on the control plane, DOS attacks) or might be simply untrusted. As long as the overhead is low, installing static rules for verifying connectivity or policies of $\domainc$ would be very valuable under these circumstances.

Once these static forwarding rules are in place, a controller can learn about the topology failures within its $\domainc$ by sending control messages from its $\domain$ into its $\domainc$ and listen to the responses.  These messages constitute control flows. Since  the controller cannot install new rules in $\domainc$, a priori static forwarding rules must be installed on $\domainc$ switches for the control flows.  Finding these static rules given the forwarding plane topology for optimal network diagnosis is the focus of this paper. Once the forwarding plane topology is learned and consistently shared across the controllers, any controller can compute these static rules, distribute to other controllers, and each controller installs them in their current $\domain$. Once the static rules are confirmed, controllers can start using them for network diagnosis.

Our system model does not assume that $\domain$ and $\domainc$ are fixed or known a priori. $C_3$ for instance might have had functional control interfaces to other switches in $\domainc$, but they may have been lost. A priori $C_3$ does not have any clue on which control interfaces would be failed. In another case, there might be a dynamic partitioning of the forwarding plane such that a controller become in charge of different forwarding elements over time based on some optimization logic. \comment{The states of different domains then could be shared between controllers as part of a control plane which might be much slower than the forwarding plane or even might be interrupted as previously pointed out.}
Not assuming fixed $\domain$ and $\domainc$ may be useful even when they are actually fixed. If multiple controllers coexist each with a different $\domain$ and $\domainc$, instead of setting static rules with respect to each controller, it would be more efficient (i.e., requires less forwarding rules for control flows) to set up static rules independent of what $\domain$ and $\domainc$ are actually. 

In contrast to $\domainc$, a controller can dynamically install new forwarding rules on its $\domain$ for control flows. As it will be clear in the later sections, we will make use of this advantage to enforce loopback of the control flows to the controller.

\begin{figure}[!t]
\begin{center}
\includegraphics[width=0.8\columnwidth]{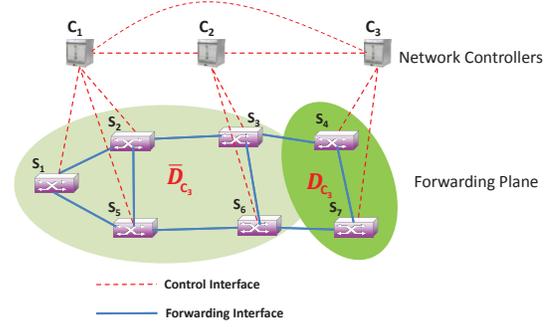}
\caption{System Model}
\label{fig:systemmodel}
\end{center}
\vspace{-0.5cm}
\end{figure}

\subsection{Failure Scenarios}
We will differentiate between two different failure scenarios: \emph{(i) Symmetric failures}, where a bidirectional link is either functional in both directions or non-functional in either direction.  \emph{(ii) Asymmetric failures}, where a bidirectional link can fail in one direction but not necessarily in the other direction.

For symmetric failures, it suffices to visit an interface in any direction to check its health. Thus, if symmetric failures are the most common scenarios, then one can specifically plan network diagnosis for such more probable incidents and save costs. We will model the forwarding plane as an undirected graph $\ugraph(\switch,\edge)$ in this case. Here, $\switch$ is the set of vertices, where there is a 1-1 mapping from the forwarding elements to the vertices in $\switch$ and  $\edge$ is the set of edges, where there is  1-1 mapping from the forwarding interfaces to the edges in $\edge$.

For asymmetric failures, both directions of the link must be examined. The problem is more constrained and imposes higher diagnosis costs. However, as we will see later in the paper, we can find optimal solutions. We will model the forwarding plane as a directed graph $\dgraph(\switch,\arc)$ in this case. Here, $\arc$ is the set of arcs, where there is  1-1 mapping from each direction of the forwarding interfaces to the arcs in $\arc$. 

\subsection{Problem Statement and Cost Metrics}
Given the forwarding plane topology and without any information on $\domain$, we would like to compute and install static forwarding rules such that  as long as controller is master of at least one switch (i.e., $|\domain| \geq 1$) it can (i) verify the topology connectivity and (ii) locate link failures.  

Static forwarding rules can be interpreted as one or more walks on the undirected ($\ugraph(\switch,\edge)$) or directed ($\dgraph(\switch,\arc)$) topology graph  of the forwarding plane. Without loss of generality, we only consider closed walks (i.e., cycles) since such walks are the only ones that can satisfy the constraint $|\domain| \geq 1$. Let $\SetW$ represent the set of all closed walks on $\ugraph$. Suppose we picked the subset of walks $\SetW_j  = \{\W_{j1}, \hdots, \W_{jK}\} \in \SetW$ to diagnose the forwarding plane. Then, we can measure the cost of $\SetW_j$ as follows.

\subsubsection{Number of control messages} Each walk is traversed by a different control packet as otherwise we cannot differentiate among the walks. Thus, the overhead in number of control messages is equal to the number of walks we picked (i.e., $K$).

\subsubsection{Latency} Each control packet $k$ traverses the walk $\W_{jk}$ and thus experiences a latency of $\tau \times \Wlenk$. Here, $\Wlenk$  is the length of walk $\W_{jk}$ and $\tau$ is the switching delay for a control packet. Depending on whether the walks can be executed in parallel or not, the overall latency figure would vary. If walks can be done in parallel, then the total latency is given by $\tau \times \max_k\Wlenk$. If walks must be done sequentially, then the total latency becomes $\tau \times \sum_k \Wlenk$. Any hybrid solution (i.e., walks are grouped together, within a group they are executed in parallel while across groups they are executed in sequence) would have a latency between these two extremes.

\subsubsection{Number of static rules} Let $\arc(W_{jk})$ represent the set of arcs traversed by $\W_{jk}$.\comment{, i.e., $\arc(W_{jk}) \subseteq \arc$.}  If an arc is traversed multiple times within the same walk or across walks, it can share the same forwarding rule at the head of arc (i.e., sending switch).\footnote{The easiest way of proving this is to let controller use source based routing. In this brute-force approach, each arc is assigned a unique label. The whole route is specified in the packet header by concatenating the labels of the links in the order they are visited. Each switch has a matching rule for the label and the action consists of popping the outermost label and forwarding. However, for data-center scale topologies, this would create very large control packets. A better way is to let switches do the packet labeling/tagging where necessary (e.g., see Section~\ref{se:mapping}).}\comment{\footnote{The easiest way of proving this is to let controller use source based routing. In this brute-force approach, each arc is assigned a unique label. Controller stacks the labels of each arc along the walk into the packet header such that the first hop label is the outermost label and the last hop label is the innermost one. Controller injects the packet to the first switch on the walk, which matches the outermost label to the specific link and takes a forwarding action that first pops the outermost label and then sends the packet to that link. Each intermediate switch does the same match-action sequence.}\textsuperscript{,}\footnote{Source based routing would lead to large control packets for data-center scale topologies. It is not hard to show that one can prevent such large control packets by instead pushing and popping labels on the forwarding plane rather than at the source by treating each duplicated sequence of arcs on a walk as a single \emph{tunnel} or \emph{trunk}. Labels are pushed before the tunnel and popped at the end of the tunnels.} }For each unique arc in any $\W_{jk}$, there must be a distinct forwarding rule as without such a rule traversing the corresponding link in the specific direction is not possible. Accordingly, we can express the total number of static forwarding rules as $|\bigcup_k \arc(W_{jk})|$, i.e., the cardinality of union of arc sets belonging to distinct walks used for diagnosis.

For an arbitrary topology, it is not possible to optimize each of these cost metrics simultaneously. In the following sections, we present solutions for topology verification and locate an arbitrary (but) single link failure. Our focus will be mainly on minimizing the total number of static forwarding rules. Our solutions are either order-optimal or optimum in this metric. We will also quantify the costs in terms of latency and number of control messages. As extensions, we will consider how adding more static rules can reduce latency and probabilistically locate multiple link failures. We will also provide how latency and number of control messages can be traded off.\comment{ without changing the cost in number of static rules.} In the final part, we exemplify why locating multiple link failures is not a solvable problem for arbitrary topologies and present simulation results about the success of locating multiple link failures using our solution.

\section{Verifying Topology  Connectivity}
\label{se:lesstrivial}
For both symmetric and asymmetric failure cases, our solutions are based on first computing a single walk $\WO$ that has the shortest length across all walks that visit each edge in $\ugraph$ or arc in $\dgraph$ at least once. For symmetric failure case, our solution is order optimal (and tight) in number of static rules. For asymmetric failures, our solution is optimum in number of static rules. For both cases, solutions are optimum in number of control messages as they require only one control message.
\subsection{Symmetric Failure Scenarios}
Computing $\WO$ is in general known as \emph{Chinese Postman Problem}. For undirected connected graphs such as $\ugraph(\switch,\edge)$, $\WO$ can be computed in polynomial time \cite{edmonds73}. Ideally, we want $\WO$ to be an Euler cycle that visits each edge in $\edge$ exactly once. A well-known necessary and sufficient condition for existence of an Euler cycle in a connected undirected graph is to have every vertex to have an even degree. \comment{In other words, all switches in the forwarding plane must have even number of interfaces.} When an Euler cycle exists, $\SetW^* = \{\WO\}$ becomes the optimum choice for the verification of topology connectivity. $\SetW^*$ minimizes both the total number of static forwarding rules and the total number of control messages. Specifically, $\SetW^*$ requires $\length{\edge}$ static rules and one control message. The latency of $\SetW^*$ becomes $\tau \times \length{\edge}$.

Unfortunately, not all forwarding topologies have Euler cycles. E.g., the forwarding plane in Fig.~\ref{fig:systemmodel} has no Euler cycles. Remember that when we want to optimize the number of static rules, it is not the length of $\WO$ (denoted as $\WOlen$), but the cardinality of $\arc(\WO)$ that must be minimized. Denote the total number of duplicate arcs for a given walk $\W$ as $\kappa(\W)$. Then, we can express the total number of static rules by $\WO$ as $(\WOlen - \kappa(\WO))$.\comment{, where $\WOlen$ is the length of $\WO$.} After stating the next lemma, we can at least claim order optimality for $\SetW^*$ in number of static rules, i.e., $\length{\edge} \leq (\WOlen - \kappa(\WO)) \leq 2\length{\edge}$. The solution is obviously optimum in number of control messages. The latency can be written as $\tau \times (\WOlen - \kappa(\WO))$.

\begin{lemma} \label{lemma:wopt} In $\WO$, no edge is traversed more than twice. In other words: $\WOlen \leq 2\length{\edge}$. 
\end{lemma}
\begin{proof}
Suppose an edge is traversed more than twice. Then, we can construct an undirected graph $\ugraph(\switch',\edge')$ from $\ugraph(\switch,\edge)$ such that $\switch=\switch'$ and $\edge \subset \edge'$ with each occurrence of an edge in $\WO$ has a 1-1 mapping to $\edge'$. Since $\WO$ is an Euler cycle on $\ugraph(\switch',\edge')$, every vertex must have an even degree as this is a necessary and sufficient condition for connected graphs 
\cite{edmonds73}.  However, if there are more than two edges between two vertices as in this case, one can remove two edges at a time between these vertices until there remains one or two edges between the same vertices. Since we started with vertices that have even degrees, removing an even number of edges would preserve the same property. In other words, an Euler path exists after these edge deletions such that it is strictly shorter than $\WO$ and visited every edge in $\edge$ at least once. Hence, $\WO$ cannot be the shortest cycle for topology verification, which is a contradiction.
\end{proof}

We can express a tighter lower bound for the number of static rules than $\length{\edge}$ by counting \emph{bridge} links. \comment{ in a given topology graph.}

\begin{definition}
An edge (if omitted) that partitions an undirected graph into two disconnected sub-graphs is called a \textbf{\emph{bridge}}.
\end{definition}

When a walk on a graph starts on one side of a bridge, if it crosses the bridge in one direction, it has to cross the same bridge in the reverse direction to come back to the starting point. Using this trivial observation, one can state the following lower bound on total number of forwarding rules.

\begin{lemma}[Lower Bound] \label{prop:lowerbound}
Topology verification requires at least $\length{\edge}  + \length{\bridge}$ forwarding rules.
\end{lemma}
\begin{proof}
For topology verification, each edge on $\ugraph$ must be traversed in at least one direction. Moreover, each bridge must be crossed in both directions as otherwise we cannot loop back to the starting point. Thus, there are at least $\length{\edge}  + \length{\bridge}$ unique arcs that must be visited. Each  uniquely visited arc requires at least one forwarding rule at the head node. Thus, we need at least $\length{\edge}  + \length{\bridge}$ forwarding rules.
\end{proof}

All the edges in tree, star, and linear topologies are bridges rendering the lower bound $2\length{\edge}$. Hence, $\SetW^*$ is indeed the optimum solution in number of static rules over such topologies.

Note that the lower bound given by Lemma~\ref{prop:lowerbound} may not be achievable by $\WO$ in general. Though, it can be achievable by a longer walk. An example is depicted in Figure~\ref{fig:counterex}. The forwarding plane topology represented by the leftmost graph has no bridges and has $\length{\edge} = 8$. Lemma~\ref{prop:lowerbound} indicates that we need at least 8 forwarding rules. Solving Chinese Postman Problem however leads to an optimal walk $\WO$ with $\WOlen = 10$. The corresponding logical ring is shown at the center of Figure~\ref{fig:counterex}. The walk traverses the arc $e_{41}$ twice, leading to $\kappa = 1$. 
Hence, $\WO$ requires $(\WOlen- \kappa) = 9$ static forwarding rules. This is strictly larger than the lower bound. The difference is due to the edge between $s_2$ and $s_3$ as it must be traversed in both directions requiring installation of one rule at $s_2$ and one rule at $s_3$. Other candidate solutions for $\WO$ suffer from a similar non-bridge link reversal. This situation is avoided over the logical ring constructed by a longer walk as shown by the rightmost ring topology in the figure. Instead of moving directly from $s_3$ to $s_2$ (as in the optimal walk), a longer path $s_3-s_4-s_1-s_2$ is taken. As a result, the ring has 12 hops with $\kappa = 4$ ($e_{12}$, $e_{34}$, $e_{41}$ occur two, two, and three times, respectively). Thus, this longer walk requires $(\Wlen-\kappa) = 8$ static rules achieving the lower bound. \comment{Next, we formalize the achievability result.}

\begin{figure}[!t]
\begin{center}
\includegraphics[width=\columnwidth]{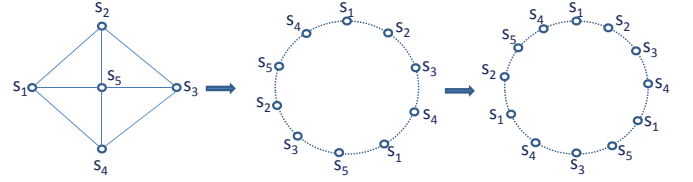}
\caption{Example of a topology, where $\WO$ cannot achieve the lower bound in Lemma~\ref{prop:lowerbound}, but a longer walk achieves it.}
\label{fig:counterex}
\end{center}
\vspace{-0.3cm}
\end{figure}

\comment{
\begin{figure}[!t]
\begin{center}
\includegraphics[width=0.5\columnwidth]{achievability.eps}
\caption{Constructive proof of achievability result using the shortest walk $\WO$.}
\label{fig:achieve}
\end{center}
\vspace{-1cm}
\end{figure}

\begin{theorem}[Achievability of Lower Bound] \label{lemma:achievability} Lower bound in Proposition~\ref{prop:lowerbound} is achievable.
\end{theorem}
\begin{proof}
We will show the achievability result by constructing a longer walk using $\WO$. If $\WO$ itself achieves the lower bound, then we are done. Otherwise $\WO$ must have at least one non-bridge edge that is traversed in opposite directions. Denote the arcs of this non-bridge edge as $e_{ij}$ and $e_{ji}$ between vertices $s_i$ and $s_j$. The goal of the construction is to replace $e_{ji}$ with a path from $s_j$ to $s_i$  by reusing the same arcs on $\WO$. This guarantees that we eliminate one link reversal without creating a new one. 

According to Lemma~\ref{lemma:wopt}, these arcs are traversed exactly once in $\WO$. Therefore, $\WO$ defines a logical ring as in Fig.~\ref{fig:achieve}, where removal of $e_{ij}$ and $e_{ji}$ divides the ring into two parts. Suppose there exists an $s_k$ that appears on both parts. Then, from $s_j$ we can trace $\WO$ up to $s_k$ in lower part of the ring (green arrows in the figure), we jump to the upper part of the ring at $s_k$ and trace $\WO$ up to $s_i$. This indeed replaces $e_{ji}$ with a longer path that reuses parts of $\WO$. Repeating this process for every non-bridge edge then would eliminate all the reverse traversals of the same edge.  Now, the proof of the problem reduces to proving that initial supposition is true, i.e., an $s_k$ that appears on both parts of the ring exists.

We can prove that the statement is true by way of contradiction. If no such $s_k$ exists, then removal of arcs $e_{ij}$ and $e_{ji}$ divide $\edge$ into two disjoint subsets with no common vertex. This implies that  the edge between $s_i$ and $s_j$ is a bridge. This is a contradiction.
\end{proof}
}

\begin{algorithm}[t]
\caption{Heuristic for reducing $(\WOlen - \kappa(\WO))$}
\label{Alg:Alg1}
\begin{algorithmic}
\STATE{\emph{Step 1:}} Find a solution $\WO$ to Chinese Postman Problem. Let $v_k$ denote the absolute position of each hop on $\WO$ and $f(v_k)$ is the actual switch at that position.
\STATE{\emph{Step 2:}} Construct set $\Lambda$ such that a pair of arcs $\{e_{ij}, e_{ji}\} \in \Lambda$ iff both $e_{ij}$, $e_{ji}$ appear in $\WO$ and the corresponding edge between $s_i$ and $s_j$ in $\ugraph$ is not a bridge.
\WHILE{$\Lambda \neq \varnothing$} 
\STATE{\emph{Step 3:}} Remove the pair of arcs $\{e_{ij}, e_{ji}\}$ from $\Lambda$ that are closest to each other on $\WO$. Without this pair, $\WO$ divides into two parts. Denote the part that keeps $s_i$ as $W_1$ and the part that keeps $s_j$ as $W_2$.
 \STATE{\emph{Step 4:}} Construct set $\Gamma$ such that $\{v_k,v_l\} \in \Gamma$ iff $f(v_k) = f(v_l)$, $v_k$ is  on $W_1$ and $v_l$ is on $W_2$.
 \WHILE{$\Gamma \neq \varnothing$}
 \STATE{\emph{Step 5:}} Remove a pair $\{v_k,v_l\}$ from $\Gamma$. 
  \STATE{\emph{Step 6:}} Denote the cycle that starts from $v_k$ on $W_1$ and ends at $v_l$ on $W_2$ as $W_3$. Denote the cycle that starts from $v_l$ on $W_2$ and ends at $v_k$ on $W_1$ as $W_4$. Construct a new walk $\WOnew$ by stitching $W_3$ to the reverse of walk $W_4$ (or equivalently stitching the reverse walk of $W_3$ to $W_4$). 
  \IF{$\kappa(\WOnew) > \kappa(\WO)$}
  \STATE{\emph{Step 7:}} $\WO := \WOnew$ and break.
  \ENDIF
  \ENDWHILE
 \STATE{Step 8:} Remove any pair  $\{e_{ij}, e_{ji}\}$ from $\Lambda$ if either of its arcs is not on $\WO$.
\ENDWHILE
\end{algorithmic} 
\end{algorithm}

\comment{
\begin{algorithm}[t]
\caption{Heuristic for reducing $(\WOlen - \kappa(\WO))$}
\label{Alg:Alg1}
\begin{algorithmic}
\STATE{\emph{Step 1:}} Find a solution $\WO$ to Chinese Postman Problem.
\STATE{\emph{Step 2:}} Construct set $\Lambda$ such that a pair of arcs $\{e_{ij}, e_{ji}\} \in \Lambda$ iff both $e_{ij}$, $e_{ji}$ appear in $\WO$ and the corresponding edge between $s_i$ and $s_j$ in $\ugraph$ is not a bridge. 
\WHILE{$\Lambda \neq \varnothing$} 
\STATE{\emph{Step 3:}} Remove a pair of arcs $\{e_{ij}, e_{ji}\}$ from $\Lambda$. Without this pair, $\WO$ divides into two cycles. Denote the cycle that starts with and ends at $s_i$ as $W_1$ and the one that starts with and ends at $s_j$ as $W_2$.
 \STATE{\emph{Step 4:}} Construct set $\Gamma$ such that $s_k \in \Gamma$ iff $s_k$ is visited by both $W_1$ and $W_2$.
 \WHILE{$\Gamma \neq \varnothing$}
 \STATE{\emph{Step 5:}} Remove an $s_k$ from $\Gamma$. 
  \STATE{\emph{Step 6:}} Locate $s_k$ on $W_1$ and $W_2$. Denote the cycle that starts from $s_k$ on $W_1$ and ends at $s_k$ on $W_2$ as $W_3$. Denote the cycle that starts from $s_k$ on $W_2$ and ends at $s_k$ on $W_1$ as $W_4$. Construct a new walk $\WOnew$ by stitching $W_3$ to the reverse of walk $W_4$ (or equivalently stitching the reverse walk of $W_3$ to $W_4$). 
  \IF{$\kappa(\WOnew) > \kappa(\WO)$}
  \STATE{\emph{Step 7:}} $\WO := \WOnew$ and break.
  \ENDIF
  \ENDWHILE
\ENDWHILE
\end{algorithmic}
\end{algorithm}
}

\begin{figure}[!b]
\begin{center}
\includegraphics[width=0.6\columnwidth]{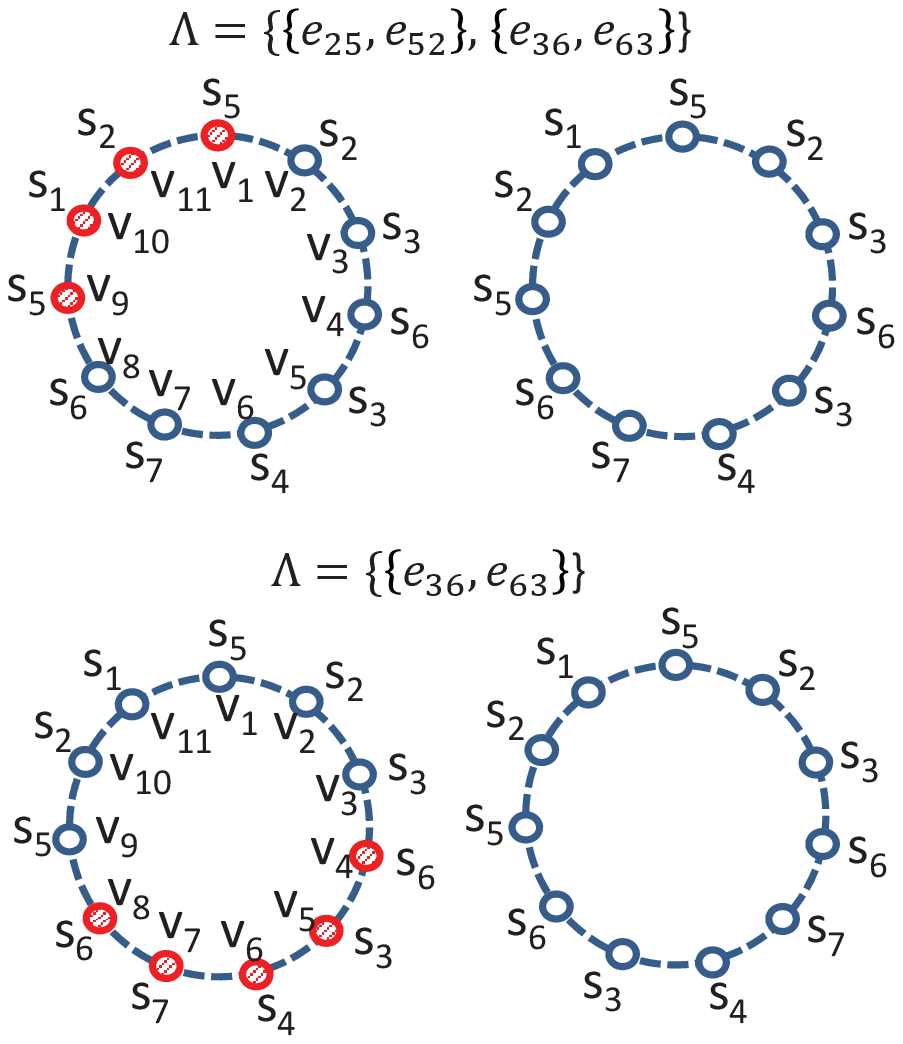}
\caption{Algorithm~\ref{Alg:Alg1} in action.}
\label{fig:alg1}
\end{center}
\vspace{-0.5cm}
\end{figure}

Although optimality in number of static rules and number of control messages can be achieved by a longer walk, we may pay a substantial penalty in delay for large topologies as the walk is much longer. Thus, it would be more desirable to pay a small penalty in number of static rules and do not incur additional delays by sticking with $\WO$. In general, there is more than one solution to Chinese Postman Problem (CPP). Then, we should search for a solution of CPP that achieves the minimum $(\WOlen - \kappa(\WO))$ (or equivalently maximum $\kappa(\WO)$)  to attain lower forwarding rule costs. Algorithm~\ref{Alg:Alg1} provides a simple transformation on top of an initial $\WO$ that iterates over the non-bridge links that have both of their arcs in $\arc(\WO)$. An example of these iterative transformations is shown in Fig.~\ref{fig:alg1} based on the forwarding plane topology given in Fig.~\ref{fig:systemmodel}. In Step 1, initial solution of CPP returns a walk of length 11 with $\kappa = 0$. In Step 2 of the algorithm, we have $\Lambda = \{\{e_{25},e_{52}\},\{e_{36},e_{63}\}\}$. In Step 3 and Step 4, we inspect the pair $\{e_{25},e_{52}\}$ and construct  $\Gamma = \{\{v_1, v_9\}\}$, respectively. In Step 6, $W_3 = s_5 \rightarrow s_1   \rightarrow s_2  \rightarrow s_5$ is reversed and the rest of the walk remains the same. The top-right ring in clockwise direction depicts newly constructed walk of same length as before but $\kappa = 1$. Hence, it requires one less forwarding rules than the initial walk. Next iteration starts with this new walk (bottom-left ring in the figure). There is only one candidate pair $\{e_{36},e_{63}\}$ to consider with $\Gamma = \{\{v_4,v_8\}\}$. By reversing the part of the ring $s_6 \rightarrow s_3 \rightarrow s_4 \rightarrow s_7 \rightarrow s_6$ in Step 6, we obtain a new walk (bottom-right ring) with $\kappa = 2$. Algorithm terminates at this stage. The newly constructed walk is still a solution of CPP. Moreover it requires $(\WOlen - \kappa(\WO)) = 9$ static rules. Since $|E|=9$ in $\ugraph$, this is an optimum walk in number of static forwarding rules.     
Our evaluations over real topologies (see Section~\ref{sec:results}) indicate that Algorithm~\ref{Alg:Alg1} is optimum in number of static rules for 60\% of the topologies, and for the rest it stays within $14\%$ of the optimum.

\subsection{Asymmetric Failure Scenarios}
Verifying topology connectivity in these scenarios require one or more control packets to visit each link in both directions. Since there are exactly $2|\edge|$ arcs to be visited and each arc must have a distinct forwarding rule, the number of static forwarding rules is $\geq 2|\edge|$. Due to our bi-directional link assumption on $\dgraph(\vertex,\arc)$, every vertex has equal in-degree (i.e., the number of incoming arcs) and out-degree (i.e., the number of outgoing arcs). In connected graphs, this is a sufficient condition for existence of an Euler cycle \cite{edmonds73}. An Euler cycle can be found in $\Theta(2|E|)$ steps over $\dgraph$. Remember that an Euler cycle visits every arc in a directed graph exactly once. Therefore, $\SetW^* = \{\WO\}$ is the optimum solution both in total number of static forwarding rules and in number of control messages. The latency of this solution, however, becomes $\tau \times 2|\edge|$.

\subsection{Mapping the Walk to Static Forwarding Rules} \label{se:mapping}
Once a walk $W$ is determined, we need to install forwarding rules at switch $s_i$ for each arc $e_{ij}$ in $\arc(W)$.  The walk and the set of corresponding forwarding rules must uniquely define a control flow for topology connectivity. For this purpose, one of the packet headers is used to identify that this packet is used to verify topology connectivity. Suppose source IP address field is used to this end and a unique local IP address $IP_A$ is assigned. Then, all static forwarding rules for $W$ must match source IP address to $IP_A$. 
Table~\ref{table:routingrules} shows how the closed walk $s_1 \rightarrow s_5 \rightarrow s_2 \rightarrow s_3 \rightarrow s_6 \rightarrow s_7 \rightarrow s_4 \rightarrow s_3 \rightarrow s_6 \rightarrow s_5 \rightarrow s_2 \rightarrow s_1$ is realized using 9 forwarding rules. Let us use this walk example below to describe the basic steps to convert a walk to a set of static forwarding rules. Note that there are many other alternatives to define matching rules and actions to realize the same walk.
\begin{table}[!t]
\caption{Sample Forwarding Rules}
\vspace{-0.3cm}
\begin{center}
\begin{tabular}{|c|c|c|}
\hline
SW & Matching Rules & Actions \\
\hline
$s_1$ & src IP ==  $IP_A$ & \parbox{2.5cm}{\centering set VLAN to $vlan_1$ \\  forward to $e_{15}$}  \\ \hline
$s_2$ & src IP ==  $IP_A$ $\wedge$ VLAN == $vlan_1$  & forward to  $e_{23}$ \\ \hline
$s_2$ & src IP ==  $IP_A$ $\wedge$ VLAN == $vlan_2$  & forward to  $e_{21}$\\ \hline
$s_3$ & src IP ==  $IP_A$ & forward to  $e_{36}$ \\ \hline
$s_4$ & src IP ==  $IP_A$ & \parbox{2.5cm}{\centering set VLAN to $vlan_2$ \\  forward to $e_{43}$}  \\ \hline
$s_5$ & src IP ==  $IP_A$ & forward to  $e_{52}$\\ \hline
$s_6$ & src IP ==  $IP_A$ $\wedge$ VLAN == $vlan_1$  & forward to  $e_{67}$ \ \\ \hline
$s_6$ & src IP ==  $IP_A$ $\wedge$ VLAN == $vlan_2$  & forward to  $e_{65}$ \ \\ \hline
$s_7$ & src IP ==  $IP_A$ & forward to  $e_{74}$
\\ \hline
\end{tabular}
\end{center}
\label{table:routingrules}
\vspace{-0.5cm}
\end{table}%

If a switch $s_i$ is visited only once by $W$, the static rule does not need to match any other packet field than the one that identifies the control packet. In our example, $s_1$, $s_4$, and $s_7$ are visited only once and these switches only need to match source IP address against $IP_A$ to take the correct forwarding action.  If a switch $s_i$ is visited multiple times but always forwards to the same link, then again the static rule does not need to match any other packet field. E.g., $s_3$ and $s_5$ occurs twice in the example walk yet their forwarding action is the same: $s_3$ forwards to $s_6$ and $s_5$ forwards to $s_2$. If a switch $s_i$ is visited multiple times and forwards to different links at some of these visits, it requires a separate forwarding rule for each of these links. In many occasions outgoing link has a one-to-one mapping to the incoming switch port and thus a forwarding rule that matches incoming switch port in addition to the source IP address is sufficient to identify the outgoing link.  When incoming switch port is not sufficient, an additional header field must be used to identify the outgoing link. One can use VLAN tagging or MPLS labeling or a custom header field using an extension to OpenFlow protocol. Suppose we use VLAN tags as it is a standard feature in most switches including OpenFlow. Then, we assign a unique VLAN tag for each outgoing link of $s_j$. In the example, $s_2$ and $s_6$ receive the same control packet twice from the same incoming interface and yet must forward to different outgoing links at each time. Since there is no 1-1 mapping to an incoming interface, incoming switch port cannot be used as a differentiating field. Consider first $s_2$. We first assign each outgoing interface a VLAN tag, $vlan_1$ to $e_{23}$ and $vlan_2$ to $e_{21}$. Table~\ref{table:routingrules} shows the matching rules with these tags. Now, the question is which switches as a forwarding action should tag the control packet. The walk example corresponds to clockwise direction over the logical ring depicted at the bottom-right corner  of Fig.~\ref{fig:alg1}. Iterating back from the tagged interface, we inspect the ring in the counter clockwise direction to identify a switch that can reuse its existing matching rule and add VLAN tagging to the action set of that matching. For instance, starting from $e_{23}$ tagged with $vlan_1$ and traversing the ring in counter clockwise direction, we first reach $s_5$. But, $s_5$ is visited twice on the ring and has only one forwarding rule. Thus, we rule it out as a candidate. Moving counter clockwise direction further, we hit $s_1$ that occurs once in the ring and has one forwarding rule. Hence, adding VLAN tagging with $vlan_1$ into its action set would not possibly contradict with another forwarding decision taken at the same switch. We repeat the process with $e_{21}$ tagged with $vlan_2$. Walking in counter clockwise direction, we reach $s_5$, $s_6$, and $s_3$ that cannot add a tagging action either because they require an additional forwarding rule ($s_5$, $s_3$) or their forwarding rules are not yet specified (e.g., $s_6$). Taking one more step counter clockwise, $s_4$ occurs only once on the ring and without modifying its matching rule we can add an additional VLAN tagging action. The only outstanding switch with no forwarding rules specified is $s_6$ at this point. We first check if we can piggyback on existing VLAN tags $vlan_1$ and $vlan_2$. Indeed, we can reuse $vlan_1$ for $e_{67}$ and $vlan_2$ for $e_{65}$ as shown in Table~\ref{table:routingrules}.

\subsection{How do controllers verify the topology connectivity?}
By constructing a single closed  walk that visits each edge (or arc) at least once, we formed a logical ring topology where all switches in the forwarding plane are part of. Thus, any controller $C_i$ can use any $s_j \in \domain_{C_i}$  to inject a control packet for topology verification.  In our example in Fig.~\ref{fig:systemmodel} and using forwarding rules in Table~\ref{table:routingrules}, $C_3$ can use $s_4$ or $s_7$ to inject a packet with its source IP address set to $IP_A$. Similarly $C_1$ can use $s_1$, $s_2$, or $s_5$ and $C_2$ can use $s_3$ or $s_6$ to inject the same control packet. Controllers must initialize the control packet header properly. E.g., if $s_2$ is the injection point, according to Table~\ref{table:routingrules}, not only the source IP address but also the VLAN tag must be assigned a valid value. In OpenFlow protocol, controllers can either tell the switch to which outgoing interface the control packet should be sent to or tell the switch to treat the packet the same as a packet coming from a particular incoming interface. In either case, the injection point is where the walk starts and once a packet is injected unless there is another rule specified, it would be indefinitely looped around the logical ring topology. Therefore, controllers must break the loop by defining a loopback rule at the injection point so that when the packet completes the walk, the packet is forwarded back to the controller that injected the packet. This loopback rule can be dynamically installed to any switch in a controller's current control domain.  Naturally, loopback rules must have priority over the static rules installed for the closed walk. Since multiple controllers might be simultaneously inspecting the topology, this loopback rule must uniquely identify the controller. A simple solution is to use destination MAC address field and install a forwarding rule at the injection point that matches this field to MAC address of the injecting controller. Thus, each controller must also set this field in the packet header before it injects it. An alternative is to use TTL field in a matching rule (i.e., check if $TTL == 0$) and set the initial value of TTL in the control packet header to the length of the walk at the controller. The action set then must include a "decrement TTL" action.

\begin{figure}[!t]
\begin{center}
\includegraphics[width=0.6\columnwidth]{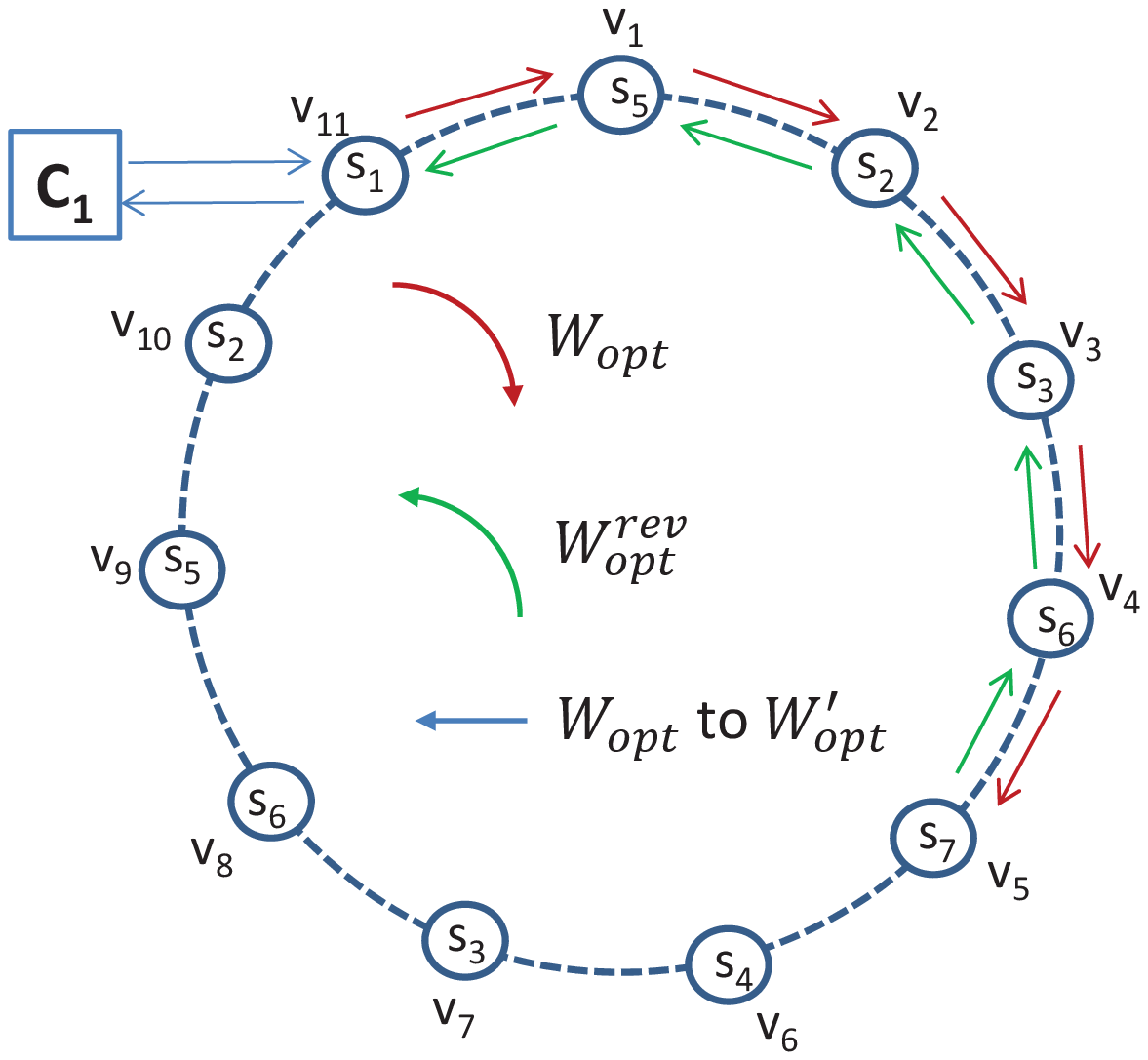}
\caption{An example of bidirectional logical ring topology. }
\label{fig:linkdetection}
\end{center}
\vspace{-0.5cm}
\end{figure}

\section{Locating an Arbitrary but Single Link Failure}
\label{se:strong}
The solution to locate an arbitrary link failure (when one or more link failures occur) is based upon constructing a bidirectional logical ring topology. At the high level, controllers inject multiple control packets each inspecting a different segment of the ring. Depending on which control messages are received back or not by the originating controller, the controller not only detects link failures but also guarantees to locate one of the failed links. 

We use the solution $\WO$ in the previous section as a clockwise walk on this logical ring. We also define a counter clockwise walk $\WOrev$ by reverting the arcs. E.g., for  $\WO = s_1 \rightarrow s_5 \rightarrow s_2 \rightarrow s_3 \rightarrow s_6 \rightarrow s_7 \rightarrow s_4 \rightarrow s_3 \rightarrow s_6 \rightarrow s_5 \rightarrow s_2 \rightarrow s_1$, $\WOrev = s_1 \leftarrow s_5 \leftarrow s_2 \leftarrow s_3 \leftarrow s_6 \leftarrow s_7 \leftarrow s_4 \leftarrow s_3 \leftarrow s_6 \leftarrow s_5 \leftarrow s_2 \leftarrow s_1$. Since we assume  bidirectional links, $\WOrev$ is a valid closed walk over both $\ugraph$ and $\dgraph$. Further, $\WOrev$ is also a solution to CPP (i.e., $\WOrev$ has length $\WOlen$) and it requires the same number of forwarding rules as $\WO$ (i.e., $\kappa(\WO) = \kappa(\WOrev)$). Once the distinct forwarding rules for $\WO$ and $\WOrev$ are installed, any controller can attach to this logical ring from any switch in its $\domain$ and inject control packets that traverse it in either clockwise or counter clockwise direction. Fig.~\ref{fig:linkdetection} shows an example of the logical ring constructed for the forwarding plane topology in Fig.~\ref{fig:systemmodel}. We label each node on the logical ring uniquely as $v_i, i=1, \hdots, \WOlen$. A switch on the forwarding plane can map to multiple logical nodes on this ring. Define $f(v_i)$  as the surjective function that maps virtual nodes on the logical ring onto the switches in the forwarding plane. In Fig.~\ref{fig:linkdetection}, actual switch IDs are shown within the circles and virtual node labels are indicated next to them. 

To be able to locate a link failure, we should be able to inspect any segment of the logical ring. For this purpose, we also install bounce back rules at each node of the logical ring. Suppose we use source IP address field to differentiate between the directions on the ring, e.g., $IP_A$ is used for $\WO$ and $IP_B$ is used for $\WOrev$. The bounce back rule at a node reverses the clockwise walk to the counter clockwise walk.  Suppose, we assign each virtual node $v_i$ a unique IP address $IP_{v_i}$.  A bounce back rule can be specified by using various fields in the packet header. Let us fix the destination IP address for this purpose. For each logical node $v_i$, a static bounce back rule is installed on $f(v_i)$ in the following form 

\emph{If source IP == $IP_A$ $\wedge$ destination $IP == IP_{v_i}$ then set source IP to $IP_B$ and send back to incoming port}.  

Note that a bounce back rule must always have a higher priority than a forwarding rule for $\WO$ when a packet has matching fields for both rules. When a controller (e.g., $C_1$ in Fig.~\ref{fig:systemmodel}) injects a packet at $f(v_{in})$ (e.g., when $v_{in} = v_{11}$, $f(v_{in}) = s_1$)  with source IP address set to $IP_A$ and destination IP address set to $IP_{v_k}$ (e.g., $IP_{v_5}$), the packet travels in clockwise direction from $v_{in}$ (e.g., $v_{11}$) to $v_k$ (e.g., $v_{5}$) and travels back to $v_{in}$ in counter clockwise direction over the constructed logical ring.  As before, the controller must install a loopback rule at injection point  $f(v_{in})$ (e.g., $s_1$) so that the switch $f(v_{in})$ forwards the packet back to the controller rather than forwarding it along $\WOrev$.

Once the static rules for $\WO$, $\WOrev$ and bounce back as well as the dynamic loopback rules are in place, each controller can fix an injection point on the logical ring and perform a binary search over the ring by eliminating half of the links from consideration at each iteration. For instance, if the link between $s_4$ and $s_7$ fails, $C_1$ can learn about this using the ring topology in Fig.\ref{fig:linkdetection} as follows. First $C_1$ determines whether all the connections are healthy or not by injecting a control packet for topology verification (e.g., source IP is set as $IP_A$). The packet never comes back to $C_1$ indicating one or more link failures. Next, $C_1$ targets half of the logical link topology by sending a control packet with source IP as $IP_A$ and destination IP as $IP_{v_5}$. The packet is received back as there are no failures in this segment. $C_1$ expands the search up to $v_8$ by setting the source IP as $IP_A$ and destination IP as $IP_{v_8}$. As this part of the ring includes the failed link, $C_1$ does not receive the packet back. $C_1$ shrinks the search up to $v_6$ and injects a fourth control packet with source IP set to $IP_A$ and destination IP set to $IP_{v_6}$. Since $C_1$ does not receive this fourth packet, but it received back the second packet, $C_1$ can conclude that the link between $s_4$ and $s_7$ has failed.

Note that when there are multiple link failures, the binary search mechanism would locate the first failure in the clockwise direction from the injection point of the logical ring. E.g., if  $e_{25}$ and $e_{47}$ fail and $C_1$ uses $v_{11}$ as the injection point, the procedure above would be able to locate only $e_{25}$.

\subsection{Cost of Locating a Single Link Failure}
Counting the rules for $\WO$, $\WOrev$, and bounce backs, the total number of static rules to be installed can be written as $3\WOlen - 2\kappa(\WO)$. Since $\kappa \geq 0$ and $\WOlen \leq 2\length{\edge}$, the total number of static rules is upper bounded by $6\length{\edge}$ for both symmetric and asymmetric failures. As locating a link failure trivially verifies the topology and topology verification requires at least $\length{\edge}$ and $2\length{\edge}$ rules for symmetric and asymmetric cases, respectively, we can easily establish the order optimality of our solution.

Including the topology verification stage, the solution requires at most $1+\lceil \log_2(\WOlen) \rceil$ control messages to be injected. Interpreting each control packet as a binary letter, $\length{\edge}$ edges and $2\length{\edge}$ arcs cannot be all checked with less than $\log_2(\length{\edge})$ and $\log_2(2\length{\edge})$ messages, respectively. Hence, we also have order optimality in number of control packets for both symmetric and asymmetric cases.

We can find the best case and worst case time delays as follows. For simplification, suppose $\WOlen$ is a power of two. The best case delay happens when each subsequent control message travels a shorter distance (i.e., exactly half of the previous one). Hence, the best case failure scenario is when the failed link occurs on the logical ring next to the injection point in the clockwise direction. Brute-force summation over these paths including the topology verification stage amounts to $(3\WOlen - 2)$ hops in total. The worst case delay happens when each subsequent control message travels a longer distance (i.e., increase exactly by half of the not inspected part of the ring). In other words, the worst case failure scenario happens when the failed link occurs on the logical ring next to the injection point in the counter clockwise direction. Again a brute-force summation over the path lengths of each control message including the topology verification stage results with total hop count of $\WOlen (2\log_2(\WOlen) - 1)  + 2$. With per hop switching delay of $\tau$, the latency $T$ is sandwiched as: \[ (3\WOlen - 2) \times \tau \leq T \leq   [ \WOlen (2\log_2(\WOlen) - 1) + 2] \times \tau \]
To put things into perspective, for data center scale operations with as much as $2^{16} = 65536$ links and $\tau = 1{\mu}s$, the lowest latency is in the order of hundreds of milliseconds and the highest latency is in the order of seconds (to be exact $196.6 ms \leq T \leq 2.032$  seconds).\footnote{We take the value of $\tau$ from http://pica8.org/documents/pica8-datasheet-48x1gbe-p3290-p3295.pdf that lists the switching latency of packets of size 64 bytes or less at $1 {\mu}s$. Our control packets are indeed short messages.}  For a moderate size topology with about 1024 links, the worst case latency becomes less than 20 ms with as low delays as 3 ms achievable.
\subsection{Speeding Up Search Time}
To reduce the latency of failure location below 1 second mark, we can trade off latency with the overhead of control messages and/or forwarding rules.

\subsubsection{More Control Messages}
Instead of just performing a sequential search on the logical ring, we can use more control packets to parallelize the search. If  we allow $m$ control message to be injected in parallel, we can inspect $(m+1)$ segments of the ring at once. We can then reduce the number of iterations to $\lceil log_{m+1}(\WOlen) \rceil$. Including the topology verification stage, the total number of control messages $M$ become $1+m \lceil log_{m+1}(\WOlen) \rceil$. In exchange, a trivial upper bound on latency ($T_{UB}$) can be expressed as $(\WOlen + 2\WOlen \lceil log_{m+1}(\WOlen) \rceil) \times \tau$. Below, we tabulate $T_{UB}$ in seconds for a large topology with $65536$ links and $\tau = 1{\mu}s$.\comment{ at interesting points.  }

\begin{center}
\begin{tabular}{|c|c|c|c|c|c|c|c|}
\hline
$m$ & 1 & 2 & 3 & 4 & 40 & 255 &65535 
\\ \hline
$M$ &  17 &   23 & 25 & 29 & 121 & 511 & 65536\\ \hline
$T_{UB}$ & 2.16 & 1.51  & 1.11   &  0.98 & 0.46 & 0.33& 0.20
\\ \hline
\end{tabular}
\end{center}

Sequential search requires 17 messages and (in previous section) we computed the worst case delay as 2.032 sec. The upper bound above for 17 messages is relatively tight at 2.16 sec.  At full parallelization with 65536 control messages all injected at the same time each checking a separate link, the upper bound on latency becomes as low as 200 ms. Using 29 messages total, we can achieve sub-second latency for locating the link failure even for this large topology. Since the control messages are short (e.g., $< $64 bytes), one can go up to m = 255 parallel messages that would incur a manageable total load (e.g., $<$16 KByte per iteration round).

\subsubsection{More Static Rules}
Instead of starting the search only in the clockwise direction, we can use both directions on the ring. To enable this, we can install bounce back rules at each $v_i$ on the ring topology to reverse $\WOrev$ onto $\WO$. For this we can assign a second unique IP address $IP_{v_i}^{(2)}$ to each $v_i$ and install a rule at each $f(v_i)$ as follows:

\emph{If source IP == $IP_B$ $\wedge$ destination $IP == IP_{v_i}^{(2)}$ then set source IP to $IP_A$ and send back to incoming port}.  

After determining which half of the ring has a faulty part, provided that the fault is closer to the injection point in the counter clockwise direction, we can switch the search direction to shorten the walk distance. For large topologies, this would cut down the worst case latency of locating a single link failure roughly by one half. Note that adding a second set of bounce back rules still preserves order optimality in number of static rules, which is now $4\WOlen - 2\kappa(\WO) \leq 8\length{\edge}$.

This reduction in latency can be combined with the results of the previous section and the upper bounds stated there can again be roughly reduced by one half for small $m$ as new trivial upper bound becomes $(\WOlen + \WOlen \lceil log_{m+1}(\WOlen) \rceil) \times \tau$.  E.g., for m = 4, $T_{UB} \approx 0.52$ seconds for $65536$ links.

\begin{figure}[!t]
\begin{center}
\includegraphics[width=0.8\columnwidth]{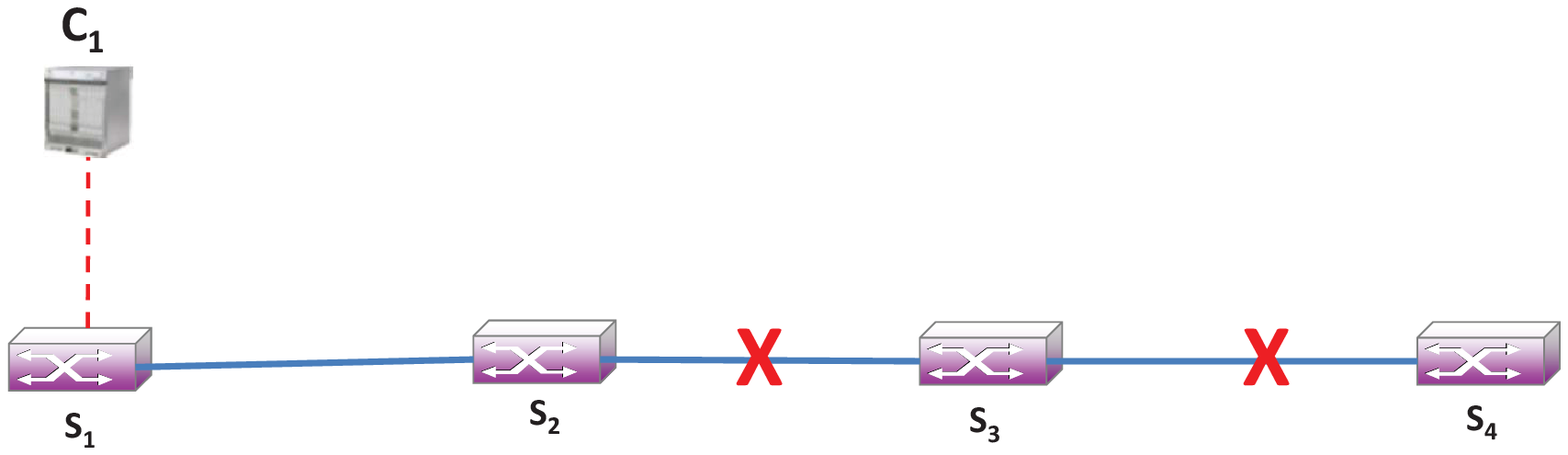}
\caption{An example of hidden undetectable failure.}
\label{fig:multifailure}
\end{center}
\vspace{-0.6cm}
\end{figure}

\section{Locating Multiple Link Failures}
\label{se:hard}
For arbitrary topologies, one cannot guarantee to locate multiple link failures. A trivial example of this is
given in Fig.~\ref{fig:multifailure}. Forwarding plane is a line topology such that a controller has only one of the end points in its $\domain$. When there are multiple failures on this topology, only the closest failure can be located and all the failures that are further away from controllerÕs $\domain$ are hidden and undetectable. E.g., the failure between $s_3$ and $s_4$ in Fig.~\ref{fig:multifailure} is not detectable by $C_1$.

Nevertheless, by building on the solution in the previous section,  we can detect more than one link failure in a probabilistic sense. First, a controller may have a switch in its $\domain$ in multiple locations on the logical ring. Further, it has typically more than one switch in its $\domain$ at any given time. Thus, it can tap on the logical ring at multiple points and locate the first failure in the clockwise direction from each of these points. If the second set of bounce back rules are also installed, it can also detect potentially other failures that are closest to each  injection point in the counter clockwise direction. Note that there are no guarantees that these detected failures map actually to the same link. Thus, with probability one, we locate at least one link failure and at non-zero probabilities up to $2\times \beta(\domain,\WO)$ failures might be located. Here, $\beta(\domain,\WO)$ counts the total multiplicity of switches in $\domain$ on $\WO$.

\section{Evaluations over Real Topologies}
\label{sec:results}
In the previous sections, we already presented upper and lower bounds on various performance metrics as well as provided some numerical values mainly for latency. For asymmetric failures, the performance can be computed easily as $\kappa(\WO) = 0$ and $\WOlen = 2\length{\edge}$.  Therefore, our evaluations mainly focus on the performance of symmetric failure case. For evaluations,
we use the public topologies posted in Internet Topology Zoo with their sizes varying from a few links to more than one hundred links. We mainly investigate two things: (i) How close do we get to the lower bounds established for symmetric failure cases? (ii) If we install bounce back rules for both $\WO$ and $\WOrev$, how many failed links do we locate? 

Fig.~\ref{fig:result1} plots the ratio $(\WOlen - \kappa(\WO))/(\length{\edge}  + \length{\bridge})$ as a function of the topology size in number of links/edges $\length{\edge}$. When the ratio is one, Algorithm~\ref{Alg:Alg1} becomes an optimum solution for topology verification. For almost $60\%$ of the topologies, indeed our solution is optimum. The ratio remains at 1.14 or below and except for 2 topologies at 1.1 or below indicating that for real network topologies we remain with at most $14\%$ of the optimum and $10\%$ of the optimum for $98\%$ of the topologies. Another observation is that the performance does not seem to be too sensitive against the topology size.

\comment{
\begin{figure}[!t]
\begin{center}
\includegraphics[width=0.8\columnwidth]{CDF_Undirected.eps}
\caption{Per Interface Cost for Symmetric Failures}
\label{fig:result1}
\end{center}
\vspace{-0.5cm}
\end{figure}
}

\begin{figure}[!t]
\begin{center}
\includegraphics[width=0.95\columnwidth]{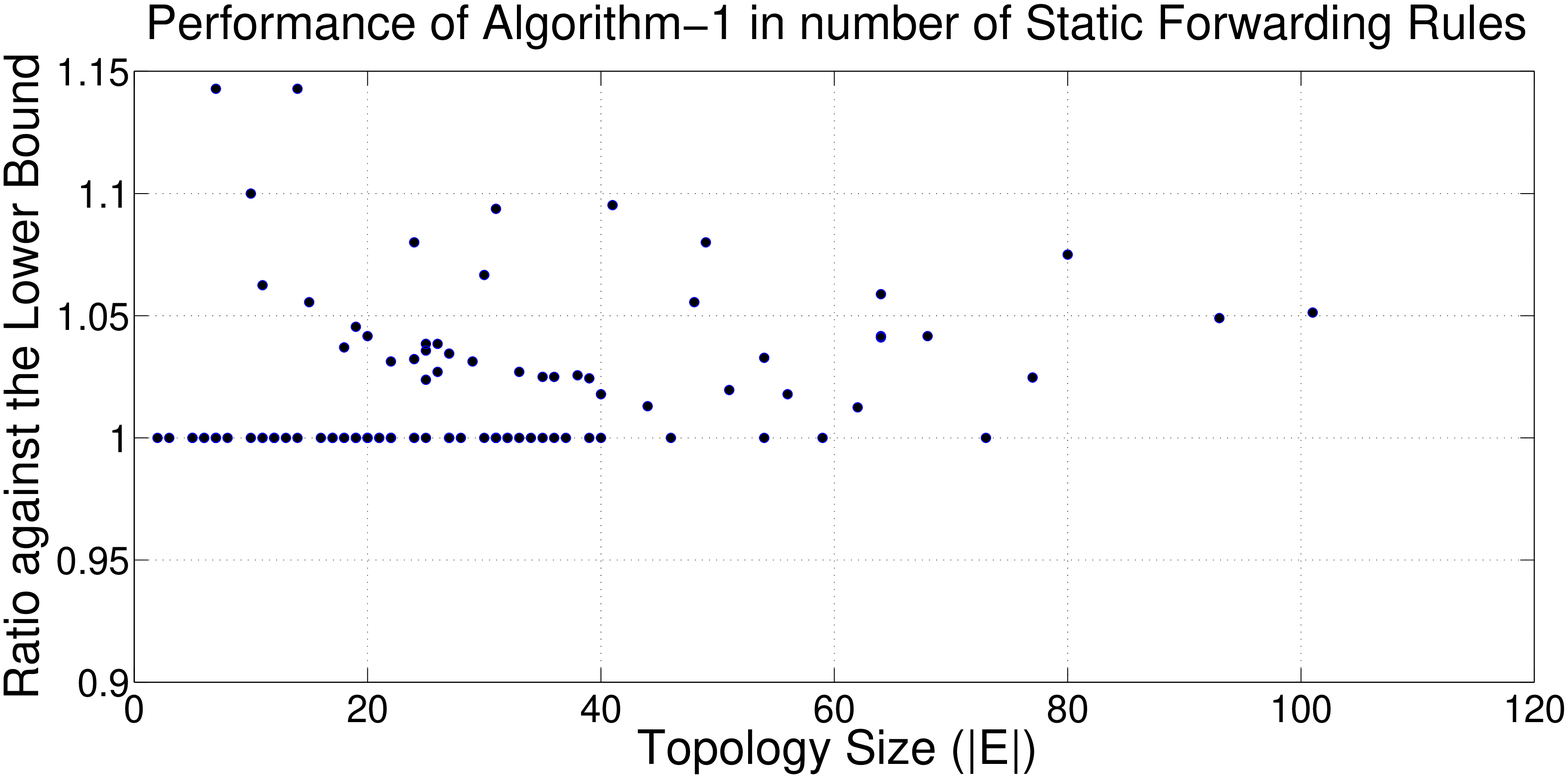}
\caption{Performance of Algorithm~\ref{Alg:Alg1} against the lower bound.}
\label{fig:result1}
\end{center}
\vspace{-0.5cm}
\end{figure}

\begin{figure}[!t]
\begin{center}
\includegraphics[width=0.95\columnwidth]{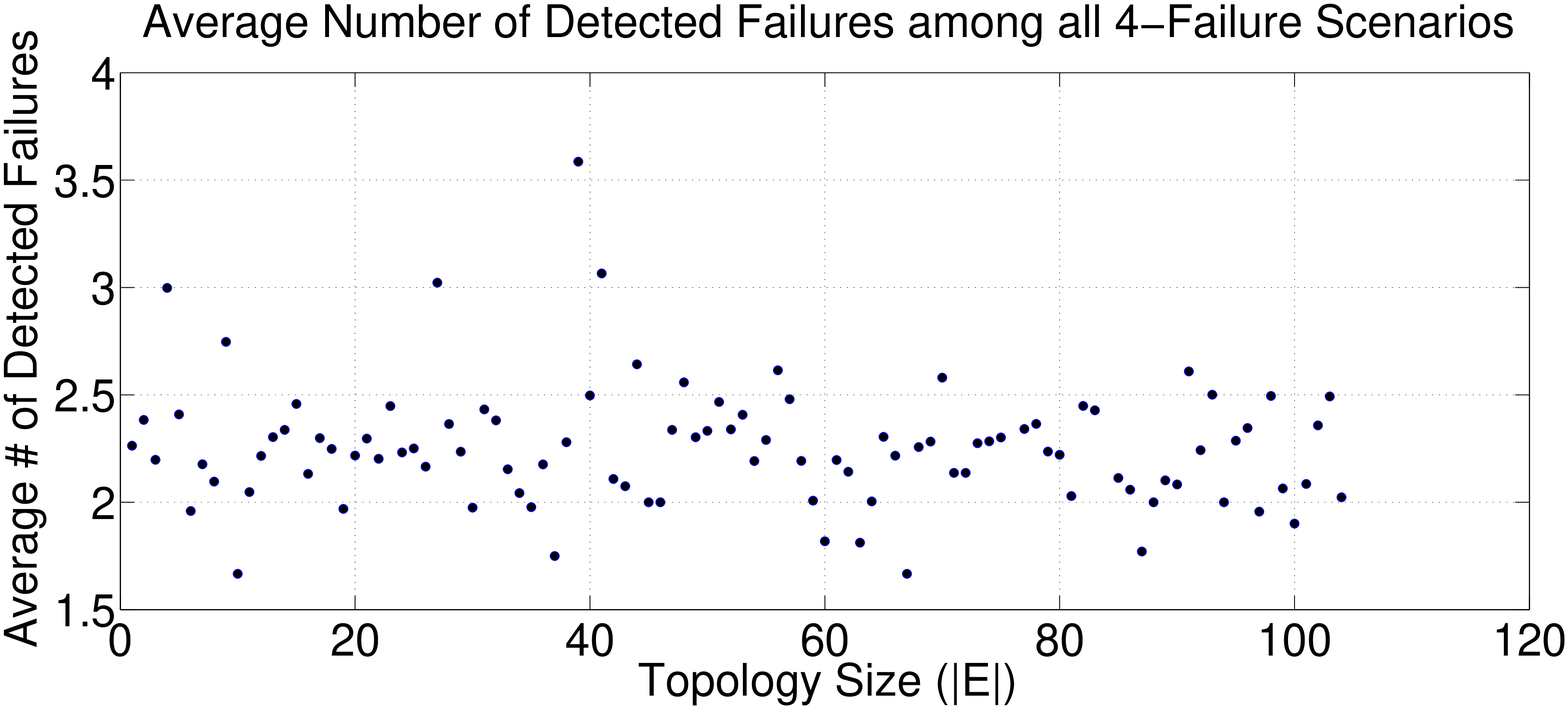}
\caption{Average number of detected failures per topology across all 4-failure patterns with $\length{\domain} = 1$.}
\label{fig:result2}
\end{center}
\vspace{-0.7cm}
\end{figure}

Fig.~\ref{fig:result2} quantifies how many failures we can actually locate. Provably, we already know that we can locate  at least one. For evaluations, we fixed the number of failures, but exhaustively iterated over every failure combination on a given topology. We then assumed that there is only one switch in $\domain$ of a particular controller and iterated over each switch as a possible injection candidate to find out as a function of injection point how many failures could be located. In the figure, we plot the average number of detected failures where the average is computed over all failure patterns and injection points for a given topology. We fixed the number of failures to four as larger number of failures was computationally quite prohibitive for us. As it can be seen, for a great majority of cases, we could actually locate two or more failures on average. For one topology with as much as 39 edges, the average was at 3.6 (i.e., for many failure patterns we could detect all four failed links). Note that although we consider only one switch in the control domain of the inspecting controller, the same switch can occur multiple times on the constructed logical ring. If a switch occurs twice, using both $\WO$ and $\WOrev$, directions we can locate up to four failures. In our evaluations, we observed that a switch can occur three times for some topologies. Clearly, not all occurrences lead to locating new failures.   We also have not observed any strong correlation between detecting more than one failure and the topology size. 

\section{Related Work}
\label{se:related}
There are several works both for classical networks and SDNs that are closely related to ours.  In all-optical networks, fault diagnosis (or failure detection) is done by using monitoring trails (m-trails) \cite{Ahuja:INFOCOM2008, Harvey:INFOCOM2007, Tapolcai:INFOCOM2009, Wu:GLOBECOM2007}. An m-trail is a pre-configured optical path. Supervisory optical signals are launched at the starting  node of an  m-trail and a monitor is attached to the ending node. When the monitor fails to receive the supervisory signal, it detects that some link(s) along the trail has failed. The objective is then to design a set of m-trails with minimum cost such that all link failures up to a certain level can be uniquely identified. Although the problem set up looks very similar, there are some fundamental differences in our work:  (1) Monitor locations are analogous to the switches in set $\domain$. Unlike all optical networking scenarios, in our problem set up, these monitor locations are not known a priori. Our solutions would work regardless of where the monitors are placed. (2) In all-optical networks, there is a per link cost measured by the sum bandwidth usage of all m-trails traversing that link. In our set up, the costs are in number of control packets and number of static forwarding rules. Furthermore, in our set up, a given static forwarding rule can be (and are in fact) reused by different walks. There are also works on graph-constrained group testing \cite{graph-group-testing:ISIT2010,graph-group-testing:TransIT2012} that is very similar to fault diagnosis in all-optical networks, and share the same fundamental differences.

SDN era has also generated many recent works on network debugging, fault diagnosis and detection, policy verification, dynamic and static state analysis \cite{Handigol:HotSDN2012,Kazemian:NSDI2012, Godfrey:NSDI2013, McGeer:ICC2012, Reitblatt:Sigcomm2012}. As far as we are aware of, all these works are complementary to our work in terms of the problem spaces they specifically target. In \cite{Kazemian:NSDI2012}, for instance, authors piggyback on existing rules installed for data flows to identify which header space can locate link failures given these rules. Installing static forwarding rules to be used for later forwarding plane diagnosis via control flows and optimizing the associated costs are the main features we have that are also absent in prior art on SDNs.

\section{Conclusion}
\label{se:conc}
We presented new results on how to diagnose forwarding plane using static forwarding rules in SDNs. Our results are provably either optimal or order optimal in number of static forwarding rules and number of control messages. For topology verification, the evaluations over real topologies revealed that our solution stayed within 14\% of the lower bound and for more than half the topologies matched the lower bound in number of control messages.  We also presented latency performance. At the expense of slight increase in bandwidth usage and forwarding rules, sub-second delays in locating link failures even at data-center scale topologies are achievable. Our solutions guarantee locating a single link failure, but can also probabilistically locate multiple link failures as dictated by the topology and failure patterns. We lack a closed form tight approximation of finding these probabilities for arbitrary topologies and it remains as a future work.









\bibliographystyle{abbrv}
\bibliography{PaperList} 


\end{document}